\documentclass[letterpaper]{article} %
\usepackage{aaai2026}  %
\usepackage{times}  %
\usepackage{helvet}  %
\usepackage{courier}  %
\usepackage[hyphens]{url}  %
\usepackage{graphicx} %
\usepackage{natbib}  %
\usepackage{caption} %
\usepackage{algorithm}
\usepackage{algorithmic}

\usepackage{newfloat}
\usepackage{listings}
\DeclareCaptionStyle{ruled}{labelfont=normalfont,labelsep=colon,strut=off} %
\floatstyle{ruled}
\newfloat{listing}{tb}{lst}{}
\floatname{listing}{Listing}
\usepackage{todonotes}

\usepackage{amsmath,amsthm,mathtools}
\newcommand{\poly}{\ensuremath{\mathsf{poly}}}

\usepackage{paralist}

\newcommand{\satisfying}[1]{\ensuremath{\mathsf{Sol}(#1)}}
\newtheorem{definition}{Definition}

\newtheorem{theorem}{Theorem}
\newtheorem{lemma}{Lemma}

\usepackage{paralist}

\title{The Limitations and Power of NP-Oracle-Based Functional Synthesis Techniques}

\author{
	Brendan Juba\textsuperscript{\rm 1} \qquad    
	Kuldeep S. Meel\textsuperscript{\rm 2,3}
}
\affiliations{\textsuperscript{\rm 1}Washington University in St.\ Louis\\
	\textsuperscript{\rm 2}Georgia Institute of Technology\\
	\textsuperscript{\rm 3}University of Toronto
}

\begin{document}

\maketitle

\begin{abstract}
	Given a Boolean relational specification between inputs and outputs, the 
	problem of functional synthesis is to construct a function that maps each 
	assignment of the input to an assignment of the output such that each tuple 
	of input and output assignments meets the specification. The past decade has 
	witnessed significant improvement in the scalability of 
	functional synthesis tools, allowing them to handle problems with tens 
	of thousands of variables. A common ingredient in these approaches is their 
	reliance on SAT solvers, thereby exploiting the breakthrough advances in SAT 
	solving over the past three decades. While the recent techniques have been 
	shown to perform well in practice, there is little theoretical understanding 
	of the limitations and power of these approaches.
	
	The primary contribution of this work is to initiate a systematic theoretical 
	investigation into the power of functional synthesis approaches that rely on 
	NP oracles. We first  show that even 
	when small Skolem functions exist, naive bit-by-bit learning approaches fail 
	due to the relational nature of specifications. We establish fundamental limitations of interpolation-based 
	approaches  proving that even when 
	small Skolem functions exist, resolution-based interpolation must produce an
	exponential-size circuits. We prove that access to an NP oracle is inherently 
	necessary for efficient synthesis.  Our main technical result shows that it is possible to use NP oracles to synthesize small Skolem functions in time polynomial  in the size of the specification and the size of the smallest sufficient set of witnesses, establishing positive results for a broad class of relational 
	specifications.

\end{abstract}

\section{Introduction}
Functional synthesis is a fundamental problem in computer science 
wherein the task is to synthesize a function that meets a given 
relational specification. Formally, let $X = \{X_1, X_2, \ldots, X_n\}$ 
be the vector of Boolean variables representing inputs and let 
$Y = \{Y_1, Y_2, \ldots Y_m\}$ be the vector of Boolean variables 
representing outputs, and let $F(X,Y)$ be a Boolean relational 
specification. Then, the problem of Boolean functional synthesis is to 
output a $m$-tuple $\Psi = \langle \psi_1, \ldots, 
\psi_m \rangle$ of Boolean functions $\psi_i$ such that 
$\exists Y F(X,Y) \equiv F(X,\Psi(X))$. Each of these functions are 
referred to as Skolem functions, and consequently, functional synthesis 
is also referred to as Skolem synthesis.

Functional synthesis has a wide range of 
applications in areas such as circuit synthesis~\cite{KS00}, program 
synthesis~\cite{SGF13}, automated program repair~\cite{JMF14}, 
cryptography~\cite{MM20}, and logic minimization~\cite{B89,BS89}. 
Given the 
computational intractability of the problem, the design of scalable 
techniques has remained the primary objective.

With the availability of powerful SAT solvers, the past two decades 
have seen a flurry of approaches to functional synthesis that rely on 
these solvers to do the \emph{heavy lifting}. These approaches have 
led to remarkable improvements in the scalability of state-of-the-art 
functional synthesis tools. Concretely,
over a standard suite of 609 benchmarks, the state-of-the-art tool in 
2016 could handle only 210 instances, while the state-of-the-art tool 
in 2023 can handle 509 instances~\cite{Golia23}. These impressive advances motivate 
the development of a foundational approach to understanding the power 
of modern functional synthesis techniques, which was also highlighted as one of the major challenges for synthesis community
\cite{AFMPS24}.

The primary contribution of this work is to initiate a systematic
theoretical investigation into the power of functional synthesis
approaches that rely on NP oracles. Our contributions are:
\begin{enumerate}
	\item We examine the applicability of computational learning theory
	to functional synthesis and demonstrate fundamental obstacles that prevent direct extension of these techniques. Specifically, we show
	that even when small Skolem functions exist, naive bit-by-bit
	learning approaches fail due to the relational nature of
	specifications and interdependencies between output variables.

\item We establish fundamental limitations of interpolation-based 
approaches for handling uniquely-defined variables. Using a 
carefully constructed example based on the pigeonhole principle, we 
prove that even when small Skolem functions exist, resolution-based 
interpolation must produce exponential-size circuits. This provides 
theoretical justification for exploring synthesis techniques beyond 
proof-based methods.

\item We prove that access to an NP oracle is inherently necessary 
for efficient synthesis even in restricted settings. This result 
suggests that the reliance on SAT solvers in practical synthesis tools 
is not merely an implementation choice but rather necessary
for achieving scalability.

\item We show that beyond uniquely defined variables, using NP oracles
Skolem functions can be synthesized in time that scales with the size
of the specification and the size of the image of the function. Moreover, 
when the function is unique, we can find a circuit of size that is polynomially 
close to the smallest possible circuit (without dependence on the specification size).
\end{enumerate}

An important question raised by our work is whether practical 
specifications, particularly those arising from program synthesis and 
repair, can be solved with such small images. If so, the 
efficiency of modern tools on these instances could potentially be 
explained by their ability to implicitly leverage this structural feature
through SAT solver-based techniques. This suggests promising directions for future work in 
analyzing the structural properties of real-world benchmarks and 
potentially developing specialized algorithms that explicitly exploit 
the structure when present.

The rest of the paper is organized as follows: we first present 
preliminaries in Section~\ref{sec:prelims}, and then discuss, in 
Section~\ref{sec:background}, state-of-the-art approaches to motivate 
the theoretical model in order to showcase the generality of the 
model studied in this paper. We then describe our results on the 
limitations of existing techniques in Section~\ref{sec:technical} and
on the power granted by NP oracles (and its necessity) in 
Section~\ref{sec:power}, and then conclude in Section~\ref{sec:conclusion}.

\section{Preliminaries}\label{sec:prelims}

Throughout this paper, we use $X = \{X_1, X_2, \ldots, X_n\}$ to refer to the set of $n$ Boolean variables representing inputs and $Y = \{Y_1, Y_2, \ldots, Y_m\}$ to refer to the set of $m$ Boolean variables representing outputs. Given a set $Z$ of Boolean variables, we use $Z^{i:j}$ to denote the set $\{Z_i, Z_{i+1}, \ldots, Z_j\}$.  We use vectors (in lowercase letters) for assignments to the Boolean variables. In particular, $\vec{x}$ and $\vec{y}$ represent assignments to $X$ and $Y$, respectively.  Given a Boolean formula $G$ over a set of variables $Z$, we say that $\vec{z}$ is a satisfying assignment if $G(\vec{z}) = 1$. We use $\satisfying{G}$ to denote the set of satisfying assignments of $G$. 
For a subset of the variables, e.g.\ $X$ with $Y=Z\setminus X$, we denote the set of satisfying assignments to $G$ projected to $X$ by $\satisfying{G}_{\downarrow X}$ (i.e., $\satisfying{G}_{\downarrow X}=\{\vec{x}:\exists \vec{y}\ G(\vec{x},\vec{y})=1\}$ where $\vec{x}$ sets the variables in $X$ and $\vec{y}$ sets the variables in $Y$).

We recalled the problem of functional synthesis in the introduction. Again: given $F(X,Y)$, output Skolem functions
$\Psi = \langle \psi_1, \ldots, \psi_m \rangle$ such that $\exists Y \, F(X,Y) \equiv F(X,\Psi(X))$ where $F(X,\Psi(X))$ represents the substitution of each $Y_i$ with $\psi_i(X)$.
From a representation perspective, it is acceptable to have $\psi_i$ depend on other $Y$ variables, as long as there are no cyclic dependencies, i.e., if $\psi_i$ (the Skolem function for $Y_i$) is defined in terms of $Y_j$, then $\psi_j$ cannot be defined in terms of $Y_i$. We will use the following notion:

\begin{definition} The 
	\emph{image} of $\Psi$, denoted $\mathrm{Im}(\Psi)$, is defined as:
	\begin{align*}
		\mathrm{Im}(\Psi) = \{\Psi(\vec{x}) : \vec{x} \in \{0,1\}^n \text{ and } 
		F(\vec{x}, \Psi(\vec{x})) = 1\}.
	\end{align*}
\end{definition}

Note that we do not require $\forall X \, \exists Y \, F(X,Y)$ to hold, i.e., it is not necessary that for every assignment to $X$, there exists an assignment to $Y$ such that $F(X,Y)$ is satisfied. For example, to model the factorization problem via functional synthesis we can write
\begin{align*}
	F(X,Y) := (X = Y^1 \times Y^2) \wedge (Y^1 \neq 1) \wedge (Y^2 \neq 1)
\end{align*}
where $Y = Y^1 \cup Y^2$ and $\times$ refers to multiplication over bit-vectors. The specification $F(X,Y)$ encodes that $Y^1$ and $Y^2$ are non-trivial factors of $X$. This specification cannot be satisfied when $X$ represents a prime number, but from a practical perspective, we are never interested in factorizing prime numbers. Therefore, we focus on factorization of only those numbers for which non-trivial factors exist.

Given a variable $Y_i$ and a set $Z \subseteq (X \cup Y \setminus \{Y_i\})$, we say that $Y_i$ is \emph{uniquely defined} in terms of $Z$ whenever
\begin{align*}
	F(X,Y) \wedge F(\hat{X},\hat{Y}) \wedge (Z = \hat{Z}) \implies Y_i = \hat{Y_i}
\end{align*}
holds, where $F(\hat{X},\hat{Y})$ refers to the formula $F$ in which $X$ (respectively, $Y$) is replaced with a freshly generated set of variables $\hat{X}$ (respectively, $\hat{Y}$).

\subsection{Resolution Proof Complexity}

Resolution-based techniques form the foundation of our lower bound 
proofs. We begin by presenting the key definitions and prior results that are relevant.

\begin{definition}
	The \emph{width} of a resolution refutation $\pi$, denoted $w(\pi)$, is 
	the maximum number of literals in any clause $\pi_i$ appearing in $\pi$. 
	Similarly, for a CNF formula $\varphi$, we denote its width by 
	$w(\varphi)$. 	The \emph{resolution width complexity} of a CNF formula 
	$\varphi$, denoted $w_{res}(\varphi)$, represents the minimum 
	width $w(\pi)$ across all possible refutations $\pi$ of $\varphi$.
\end{definition}

Ben-Sasson and Wigderson~\shortcite{BW2001} established a key connection between refutation length and width:

\begin{theorem}
	\label{thm:length-width}
	For any unsatisfiable CNF formula $\varphi$,
	$L_{res}(\varphi)\geq 2^{\Omega((w_{res}(\varphi)-w(\varphi))^2/n)}$,
	where $L_{res}(\varphi)$ denotes the minimum length over all resolution refutation of $\varphi$.
\end{theorem}

To establish lower bounds on refutation width, we employ the Pudl\'{a}k-Buss 
\shortcite{PB1995} game framework: 	The \emph{formula game} involves two players: the Prover and the Liar. The Liar claims that formula $\varphi$ is satisfiable, while the Prover 
	attempts to demonstrate unsatisfiability. The game proceeds in rounds, 
	with the Prover selecting a variable each round and challenging the Liar 
	to assign it a value. The Prover maintains a record of previous 
	assignments, and may selectively forget them. The Prover's variable 
	selection depends solely on the current round and its memory contents. 
	The Prover wins if the remembered partial assignment $\vec{\rho}$ 
	falsifies any clause in $\varphi$. The Liar wins by maintaining a 
	strategy avoiding such falsifying assignments.

		The connection between game strategies and resolution width is captured by 
		the following theorem:
		
\begin{theorem} {\bf \cite{PB1995,P2000}}
			\label{thm:width-game}
			If a resolution refutation of $\varphi$ with width $w$ exists then
			the Prover has a winning strategy while remembering at most $w+1$ 
			variables simultaneously. Conversely, if the Prover can win while 
			remembering at most $w$ variables, there exists a resolution refutation 
			of width $w$.
\end{theorem}

\section{Background}\label{sec:background}
We now recall the necessary theoretical foundations for our 
work. We begin by reviewing functional synthesis, focusing on key 
computational complexity results and examining the major paradigms 
that have emerged in state-of-the-art approaches. We then review
relevant concepts from computational learning theory, particularly 
mistake-bounded learning, which provides crucial insights for our 
theoretical analysis of synthesis techniques that leverage NP oracles.

\subsection{Functional Synthesis}

These state of the art approaches for functional synthesis seek to harness the power of SAT solvers.  
We can largely classify these approaches into  four paradigms: {\em proof-based techniques},  {\em knowledge-compilation-based},  {\em guess-check-repair}, and {\em incremental determinization}. 
The {\em proof-based techniques} primarily focus on the specification $F(X,Y)$ for which $\forall Y \exists X F(X,Y)$ is true~\cite{BJ12,NPLSB12,HSB14,RT15,BJJW15,SSWZ20}. In such cases, the Skolem functions for the corresponding $y_i$ can be derived from the proof of validity. To generate the proof of validity for $\forall Y \exists X F(X,Y)$, these techniques rely on SAT-based Quantified Boolean Formula (QBF) solvers that extend conflict-driven clause learning techniques to QBF settings.

The {\em guess-check-repair} paradigm traces its roots to early efforts
 that were inspired by the success of CEGAR (Counter-Example Guided Abstraction Refinement) approaches in formal verification~\cite{JSCTA15,ACJS17,ACGKS18}. The underlying idea is to {\em guess} with candidate functions.  John et al.~\shortcite{JSCTA15} observed that  for a given relational specification $F(X,Y)$, checking whether $\hat{\Psi}$ is a Skolem function reduces to satisfiability of the following formula, referred to as the Error Formula:
\begin{align*}
	E(X,Y,Y') := F(X,Y) \wedge \neg F(X,Y') \wedge (Y' \leftrightarrow \hat{\Psi}(X))
\end{align*}

Observe that $E(X,Y,Y')$ is unsatisfiable if and only if $\hat{\Psi}$ is a Skolem function vector. Furthermore, if $E(X,Y,Y')$ is satisfiable, then the satisfying assignment $\sigma$ of $E$ represents a {\em counterexample}, which necessitates repair. Ideally, we would like the repair to {\em generalize}, i.e., also fix other potential counterexamples. To this end, the state-of-the-art techniques rely on UNSAT cores to construct sound repairs that generalize. The check-repair loop continues until the error formula $E(X,Y,Y')$ becomes unsatisfiable. The use of data-driven approaches for the {\em guess} step has led to significant scalability gains, as evidenced by the performance of the state-of-the-art synthesis engine, $\mathsf{manthan}$~\cite{GRM20,GSRM21}.

The compilation-based paradigm relies on the observation that the complexity of functional synthesis techniques depends on the representation of $F(X,Y)$~\cite{acs23}. The earliest works were inspired by the observation that when $F$ is represented as an Ordered Binary Decision Diagram (OBDD)~\cite{JSCTA15,FTV16}, the Skolem functions for $Y_i$ can be synthesized in polynomial time. Subsequent work sought to identify a broader class of representations that would allow for the polynomial-time synthesis of Skolem functions~\cite{ACGKS18,ACGKS21}. 

The {\em incremental determinization} approach is based on the observation that it is easy to extract Skolem functions for variables that are uniquely defined and accordingly iteratively adds additional clauses to $F(X,Y)$ so that every variable is uniquely defined~\cite{RS16,RTRS18}. 
 The state-of-the-art method, \textsc{Cadet}, relies on lifting the conflict-driven clause learning framework proposed in the context of SAT solving to functional synthesis, which entails many invocations of the SAT solver.

To summarize, these diverse approaches have  enabled significant advances in the scalability of state-of-the-art methods, as discussed in the introduction.
Theoretical studies in the context of functional synthesis, meanwhile, have primarily focused on the hardness of the problem. Such studies do little to explain the impressive advances achieved over the past decade. In particular, a natural question from a practitioner's perspective is to understand the power of frameworks that rely on SAT solvers. We seek to address this gap.

\subsection{Computational Learning Theory}

Our work draws from work on machine learning in the
\emph{mistake-bounded} learning model first introduced by B\={a}rzdi\c{n}\v{s} 
and Freivalds~\shortcite{BF72} (although such a model was also used in the
analysis of the Perceptron~\cite{Rosenblatt1958}). Whereas most familiar models of learning 
such as PAC learning~\cite{Valiant1984} are statistical, mistake-bounded 
learning is a worst-case model that  demands that the learner learn to
correctly evaluate the target function on all inputs. Because of this crucial
difference, the model is relevant to the synthesis task, as we will discuss in
more detail later.

In the mistake-bounded model, a learning problem is given by
a class of Boolean functions $\mathcal{C}$ mapping from $\{0,1\}^n$ to $\{0,1\}^m$. 
In an instance of the problem, some function $c^* \in \mathcal{C}$ 
is fixed, followed by an interaction between an algorithm, the \emph{learner}, 
and an oracle, the \emph{environment}. Each round proceeds as follows:
\begin{compactenum}
	\item The environment chooses $\vec{x} \in \{0,1\}^n$ arbitrarily and provides it to the learner as input.
	\item The learner, based on its previous state and the environment's chosen input, provides a \emph{prediction} $\vec{y} \in \{0,1\}^m$ to the environment.
	\item The environment provides $c^*(\vec{x})$ to the learner, which chooses a state for the next round.
\end{compactenum}
If $\vec{y} \neq c^*(\vec{x})$, then we say that the learner has made a \emph{mistake}.
An algorithm has a \emph{mistake bound} $M$ if for all $c \in \mathcal{C}$ and all
possible infinite sequences of inputs the environment may choose, the learner
makes at most $M$ mistakes; hence, for all sequences, there is some finite
round $t$ after which the learner correctly predicts $\vec{y} = c^*(\vec{x})$. Typically
we are interested in \emph{conservative} learners, for which the learner's
state only changes following a mistake (and otherwise, the learner uses an
identical state for the subsequent round). Note that once the learner's state
is fixed, it computes a fixed function $c: \{0,1\}^n \to \{0,1\}^m$, and there 
is a circuit that makes the same predictions as the learner's algorithm on this state.

B\={a}rzdi\c{n}\v{s} and Freivalds~\shortcite{BF72} also introduced the \emph{halving} strategy for Boolean functions, for which the learner predicts
according to a majority vote of all functions $c \in \mathcal{C}$ that satisfy
$c(\vec{x}) = c^*(\vec{x})$ for all $\vec{x}$ observed on previous rounds. They
observed that this method obtains a mistake bound of $\log|\mathcal{C}|$, but
they did not consider computational aspects of the model. By contrast, 
Littlestone~\shortcite{Littlestone1988} introduced such considerations, and gave a 
polynomial-time algorithm for learning halfspaces with polynomially-bounded 
coefficients. (He also introduced methods for analyzing mistake bounds.) Angluin~\shortcite{Angluin1988} later observed that the mistake-bounded model was equivalent 
to a model in which the learner has a \emph{counterexample} oracle: here, again, 
a function $c^* \in \mathcal{C}$ is fixed, and the computation proceeds as follows:
\begin{compactenum}
	\item The learner sends a representation of some $c: \{0,1\}^n \to \{0,1\}^m$ to the oracle.
	\item The oracle either indicates that $c(\vec{x}) = c^*(\vec{x})$ for all $\vec{x} \in \{0,1\}^n$ (learning is successful) or else provides a \emph{counterexample} $\vec{x}^*$ such that $c(\vec{x}^*) \neq c^*(\vec{x}^*)$, chosen arbitrarily.
\end{compactenum}
We here measure the number of rounds before the learner's function is correct.
It is easy to see that the learner here can send a representation of a mistake-bounded learner's function on its current state and obtain a next input on which the mistake-bounded learner errs, so the number of rounds of interaction here is at most the mistake bound. Conversely, a mistake-bounded learner can simulate the oracle by continuing to predict according to the function $c$ provided by a learner in the oracle model. In this form, the connection to counterexample-guided synthesis is more apparent.

 Bshouty et al.~\shortcite{BCGKT1996}  showed how to simulate the halving method of 
B\={a}rzdi\c{n}\v{s} and Freivalds~\shortcite{BF72} using the aid of an NP oracle,
to learn Boolean circuits. This is the starting point for our work, and we will
review it in detail next. It is worth emphasizing the two key differences between synthesis problem and  the mistake-bounded learning model: Firstly, we
are given a formula that captures the behavior of the oracle, whereas in
mistake-bounded learning the oracle is a black box to the learner. Secondly, we are particularly interested in functions with more than a
single bit of output, where it is not immediately obvious how to generalize the
majority vote strategy.

\section{Limitations of Existing Approaches}\label{sec:technical}

In this section, we focus on analyzing limitations of existing approaches. To this end, we demonstrate inherent limitations for two classes of techniques: synthesizing output variables one variable at a time and interpolation-based approaches. 
\subsection{Limitations of Sequential Algorithms }
\label{subsec:learning}

A natural approach to functional synthesis is synthesizing Skolem functions for 
individual output variables sequentially. This strategy appears promising 
given the success of computational learning theory techniques for single-output 
Boolean circuits. However, we demonstrate that such bit-by-bit approaches 
fundamentally fail for multi-output synthesis, even when small Skolem 
functions are guaranteed to exist.

\begin{definition}
	Let $\mathcal{C}$ be a set of Boolean circuits with $n$ inputs 
	$X_1,\ldots,X_n$ and $m$ outputs, and let $\delta\in (0,1/2]$. We 
	say that a function $h:\{0,1\}^n\to\{0,1\}^m$ is 
	\emph{$\delta$-good} for $c\in\mathcal{C}$ if for any 
	$x\in\{0,1\}^n$ such that $h(x)\neq c(x)$ (a \emph{counterexample} 
	to ``$h=c$''), 
	$|\{g\in\mathcal{C}:g(x)\neq c(x)\}|\geq \delta |\mathcal{C}|.$
\end{definition}

For circuits with a single output ($m=1$), Bshouty et al.~\shortcite{BCGKT1996} proved that 
sampling $O(n)$ circuits approximately uniformly at random from those 
consistent with prior counterexamples and using a majority vote suffices to 
construct a good hypothesis. This approach yields polynomial-time learning 
algorithms for Boolean circuits when combined with NP oracle access.

The core difficulty lies in the relational nature of functional synthesis 
specifications. Unlike classical function learning where each input has a 
unique corresponding output, relational specifications permit multiple valid 
outputs for a given input. When multiple outputs are valid, the choice of 
assignment for early variables determines the complexity of synthesizing 
functions for later variables.

\begin{definition}
	A \emph{natural sequential synthesis algorithm} constructs Skolem functions for 
	$Y_1, \ldots, Y_m$ in order, where for each $Y_i$, it samples candidate 
	functions uniformly at random from those consistent with prior counterexamples 
	and uses majority vote to select the hypothesis, following the approach of 
	Bshouty et al.
\end{definition}

\begin{theorem}\label{thm:sequential-failure}
	There exists a family of relational specifications $\{R_m\}_{m \geq 4}$ such that 
	\begin{enumerate}
		\item $R_m$ has Skolem functions computable by circuits of size $O(nm)$
		\item Any natural sequential synthesis algorithm that constructs Skolem functions for 
		$Y_1, \ldots, Y_m$ requires circuits of size $2^{\Omega(m)}$ with 
		probability $1 - 2^{-\Omega(m)}$
	\end{enumerate}
\end{theorem}

\begin{proof}
	We construct $R_m$ where $Y$ is partitioned into two blocks of size $m/2$ 
	each. Let $\vec{s} \in \{0,1\}^{m/2}$ be a fixed string, and define functions 
	$c:\{0,1\}^n \to \{0,1\}^{m/2}$ computable by circuits of size $O(n)$ and 
	$h:\{0,1\}^n \times \{0,1\}^{m/2} \to \{0,1\}^{m/2}$ that requires circuits 
	of size $2^{\Omega(m)}$. Set
	\begin{align*}
		R_m = &\{(\vec{x},(\vec{s},c(\vec{x}))):\vec{x}\in\{0,1\}^n\}\cup\\
		&\{(\vec{x},(\vec{y},h(\vec{x},\vec{y}))):\vec{x}\in\{0,1\}^n,
		\vec{y}\in\{0,1\}^{m/2}\setminus\{\vec{s}\}\}
	\end{align*}

\noindent	
	First, we verify that $R_m$ has small Skolem functions. Define:
	\begin{align*}
		\psi_i(X) &= s_i \text{ for } i = 1, \ldots, m/2\\
		\psi_i(X) &= c_{i-m/2}(X) \text{ for } i = m/2+1, \ldots, m
	\end{align*}
	These functions have circuit size $O(nm)$ and satisfy $R_m$.
	
	Now consider any natural sequential synthesis algorithm $\mathcal{A}$ that constructs 
	functions for $Y_1, \ldots, Y_{m/2}$ before constructing functions for 
	$Y_{m/2+1}, \ldots, Y_m$. When $\mathcal{A}$ synthesizes the first $m/2$ 
	variables, every assignment appears valid since the relation guarantees some 
	completion exists.
	
	Since $\mathcal{A}$ follows Bshouty et al.'s approach, it samples candidate 
	functions uniformly at random from those consistent with prior counterexamples 
	and uses majority vote. For the first $m/2$ variables, all possible assignments 
	are consistent with the specification (as each has a valid completion), so the 
	majority vote over uniformly sampled candidates will select assignment 
	$\vec{t} \in \{0,1\}^{m/2}$ with probability $2^{-m/2}$ for $\vec{t} = \vec{s}$ 
	and probability $1 - 2^{-m/2}$ for $\vec{t} \neq \vec{s}$.
	
	If $\vec{t} \neq \vec{s}$, then for the remaining variables $Y_{m/2+1}, \ldots, Y_m$, 
	the algorithm must output functions that compute $h(\vec{x}, \vec{t})$. By 
	construction of $h$, this requires circuits of size $2^{\Omega(m)}$.
	Therefore, with probability $1 - 2^{-m/2}$, any natural sequential synthesis 
	algorithm produces Skolem functions requiring exponential circuit size.
\end{proof}

This theorem demonstrates that natural sequential synthesis approaches fail even when 
small Skolem functions exist. The fundamental issue is that local decisions 
made early in the synthesis process can render later synthesis steps 
exponentially difficult, despite the existence of a globally optimal solution 
with small circuits.

\subsection{Limitations of Interpolation-Based Approaches}

While interpolation-based methods have proven successful for 
certain classes of specifications, we demonstrate fundamental 
limitations of this approach. We present a simple example based on 
the pigeonhole principle where small Skolem functions exist, yet 
resolution-based interpolation necessarily produces 
exponential-size circuits. Our example will not be one with unique
Skolem functions, but recall that our primary interest here is to what extent 
we can extend beyond unique Skolem synthesis.

Recall that resolution is a logic on the language
of clauses, where there is a single rule of inference that allows
inferring a clause $C\lor D$ from clauses $x\lor C$ and $\neg x\lor D$.
Slivovsky's interpolation-based method for unique synthesis \cite{F20} relies on
feasible interpolation (essentially formulated by Kraj{\'\i}{\v{c}}ek \shortcite{K1994}):
\begin{definition}
We say that a proof system has \emph{feasible interpolation} if there is a polynomial $p$ such that for any pair of formulas $\varphi_0(A,C)$ and $\varphi_1(B,C)$ on $n$ variables $A,B,C$ such that $\varphi_0(A,C)\land \varphi_1(B,C)$ has a refutation $\pi$ of size $s$, there is a circuit $I(C)$ of size $p(n,s)$ called an \emph{interpolant} such that for any assignment $\vec{a},\vec{b},\vec{c}$, if $I(\vec{c})=0$ $\phi_0(\vec{a},\vec{c})$ is false and if  $I(\vec{c})=1$, $\varphi_1(\vec{b},\vec{c})$ is false.
\end{definition}
Slivovsky uses the feasible interpolation of resolution to obtain a circuit for the $i$th bit of $Y$ given circuits for $i+1,\ldots,n$ as follows: we take $A$ and $B$ to be $Y^{i+1:n}$, $C$ to be $X,Y^{1:i-1}$, and both formulas are constructed from $\neg F$ with $Y^{i+1:n}$ given by our previously constructed Skolem functions. $\varphi_0$ then additionally fixes $Y_i$ to $0$ and $\varphi_1$ fixes $Y_i$ to $1$. Then the interpolant circuit for $Y_i$ gives a setting such 
that when $X=\vec{x}$ and $Y^{1:i-1}=\vec{y}^{1:i-1}$, fixing $Y_i=I(\vec{x},\vec{y}^{1:i-1})$ gives $\neg F$ is false---i.e., $F$ is satisfied.
The size of the smallest resolution refutation is a lower
bound on the size of the circuit constructed by the feasible
interpolation construction. Thus, we can show an exponential lower bound
on the size of the circuit we obtain by showing an exponential lower
bound on the size of the smallest resolution refutation of the pair 
of formulas constructed by Slivovsky's method.

To prove the lower bound, we consider the following version of the
pigeonhole principle~\cite{FLN+2015} $\text{bPHP}^k_n$ where $n=km$ for $m=\lceil\log_2n
\rceil-1$:
\begin{definition}
	Let $X$ represent $k$ blocks of $m$ bits (hole addresses), denoted $X_{i,j}$ for
	$i=1,\ldots,k$, $j=1,\ldots,m$. Using the notation $\ell_{0}(X)=\neg X$
	and $\ell_1(X)=X$, define:
	\begin{equation*}
\small
	F(X,Y)=\qquad\mathclap{\bigvee_{\substack{\vec{b}\in\{0,1\}^m,\\ 
				1 \leq i_1 < i_2 \leq k
	}}}\qquad
	\bigwedge_{j=1}^m\ell_{\vec{b}_j}(X_{i_1,j})\land \ell_{\vec{b}_j}(X_{i_2,j})\land
	\ell_{\vec{b}_j}(Y_j).
\end{equation*}
\end{definition}
\noindent
Intuitively, assignments to $X$ specify assignments to one of $2^m$ holes for each of our $k$ pigeons, and in satisfying assignments, $Y$ indicates a hole containing at least two pigeons.

\begin{theorem}\label{thm:php-width}
	Any resolution refutation of Slivovsky's first interpolation formula for
	$\text{bPHP}^k_n$ (where $m <\log_2 n$) has width at least $2^m-1$.
\end{theorem}
\begin{proof}
Fix a strategy for the Prover in which the Prover remembers fewer than $2^m$ variables in any state. We say that the partial assignment \emph{mentions} a pigeon $i$ if it includes some variable $X_{i,j}$ for some $j$.  The Liar can now win using the following strategy. 
Inductively, the Liar will remember a set of distinct hole assignments for the pigeons mentioned in the Prover's current memory; when the Prover forgets all of the variables for a given pigeon, the Liar also forgets the assignment chosen for that pigeon. This allows the Liar to continue to maintain this assignment and answer the Prover's queries:

Initially, the Prover's memory is empty and so the empty set of assignments suffices. Whenever the Prover queries a new variable, if it is for a pigeon that is mentioned in the Prover's current memory, the Liar answers with the corresponding bit of the assignment currently chosen for that pigeon. Otherwise, for a new pigeon, since the Prover's state includes at most $2^m-1$ variables, the Liar currently has assigned at most $2^m-1$ holes, so the Liar can choose a new index distinct from those chosen for the currently mentioned pigeons. 
(The $Y$ variables may be set arbitrarily.)

Now, since at any point the Prover's memory contains assignments consistent with the Liar's chosen set of distinct holes for the mentioned pigeons, none of the clauses of $\text{bPHP}^k_n$ are falsified. Since this first formula only includes clauses from $\text{bPHP}^k_n$, the Liar wins. Since by Theorem~\ref{thm:width-game}, a width complexity of $2^m-2$ would give the Prover a winning strategy remembering fewer than $2^m$ variables, the width complexity of the formula is at least $2^m-1$.
\end{proof}

We are now ready to state our lower bound:

\begin{theorem}\label{thm:php-hard}
	There exists a formula family with Skolem functions of size
	$\tilde{O}(n^4)$ where resolution-based interpolation produces circuits
	of size at least $2^{\Omega(n/\log^2 n)}$.
\end{theorem}
\begin{proof}
	First, we show that $\text{bPHP}^k_n$ has small Skolem functions. Observe,
	a Skolem function may be computed by setting $Y_i$ as in the lexicographically
	first collision, i.e.,
	\begin{align*}
		f_i(X)=&\bigvee_{{\substack{\vec{b}\in\{0,1\}^m\\\vec{b}_i=1}}}
		\left(
		\bigwedge_{j=1}^m\ell_{\vec{b}_j}(X_{i_1,j})\land \ell_{\vec{b}_j}(X_{i_2,j})\qquad\land \right.\\
		&\left.\bigwedge_{{\substack{\vec{b}'\in\{0,1\}^m:\vec{b}'<\vec{b}\\ 1\leq i_1<i_2 \leq k}}} \, \, \, \bigvee_{j=1}^m \ell_{1-b'_i}(X_{i_1,j})\lor \ell_{1-b'_j}(X_{i_2,j}) \right).
	\end{align*}
	Here, by construction, $2^m<n$, and therefore, the circuit has size $\tilde{O}(n^4)$.
	Note that the $\text{bPHP}^k_n$ formulas have width $3m$, which dominates the width of the formulas $\varphi_0$ and $\varphi_1$ constructed by Slivovsky's method. By Theorem \ref{thm:php-width}, we have an exponential lower bound on
	resolution width for refuting the pair. Theorem~\ref{thm:length-width}
 gives an exponential proof size bound, which in turn gives the claimed bound on the circuits extracted from the proof by Slivovsky's method.
\end{proof}

\section{The Power of NP Oracles}\label{sec:power}
We now turn to solving synthesis using NP oracles, inspired by Bshouty et al.~\shortcite{BCGKT1996}. We first note that the NP oracles are really necessary for efficient Skolem synthesis in general.

\subsection{Necessity of NP Oracle Access}

A fundamental question is whether the full power of an NP oracle is
truly necessary for efficient synthesis. Satisfying assignments are a
special case of Skolem functions, when there are no universally quantified
variables. Then, the condition that the variables are uniquely defined
corresponds to a unique satisfying assignment. It is then a straightforward
corollary of the classic reduction by Valiant and Vazirani~\shortcite{valiant1986np}
from SAT to unique-SAT that unique Skolem synthesis suffices to solve satisfiability:

\begin{theorem}
	If there is a randomized polynomial-time algorithm such that on input a formula
	$F(X,Y)$ such that $\forall X\exists Y F(X,Y)$ is a tautology in which $y_i \in 
	Y$ is uniquely defined in terms of $(X,Y^{1:i-1})$ by a circuit of size polynomial in $|F|$, returns a
	Skolem function for $y_i$, then NP$=$RP.
\end{theorem}
\begin{proof}
	Recall that Valiant and Vazirani~\shortcite{valiant1986np} gave a polynomial-time 
	reduction from CNF satisfiability that
	\begin{compactenum}
		\item if the input CNF is satisfiable, produces a CNF with a unique satisfying assignment with constant probability
		\item if the input CNF is unsatisfiable, produces an unsatisfiable CNF
	\end{compactenum}
	Suppose we run the reduction and obtain the CNF $\Psi(Y)$; if the original formula
	was satisfiable and $\Psi$ had a unique satisfying assignment, on input 
	$\exists Y\Psi(Y)$, the hypothetical Skolem synthesis algorithm must run in time 
	polynomial in the size of $\Psi$ and output constant circuits (no inputs) for each 
	$y_i$ that satisfy $\Psi(Y)$, i.e., which evaluate to the unique satisfying 
	assignment of $\Psi$. This gives an RP algorithm for SAT and hence NP.
\end{proof}

Thus, in conclusion, we see that an algorithm for unique Skolem synthesis can essentially be used as an NP oracle. Deciding NP is thus a necessary condition for this problem.

\subsection{When NP Oracles Suffice}

Now that we see that it is not enough to directly apply the method of Bshouty 
et al.\ bit-by-bit, we turn to identifying situations where we can extend the
method to learn Skolem functions. The first case is when $Y$ is unique. 
Actually, more generally, we can find circuits for those variables in $Y$ that are uniquely
determined by previous variables.

\begin{theorem}\label{lem:unique-bits}
	For a given $F(X,Y)$, if $Y_i \in Y$ is uniquely defined in terms of 
	$(X,Y^{1:i-1})$ by a circuit of size $s$, then we can learn a Skolem 
	function for $Y_i$ with polynomially many NP oracle calls with high probability.
\end{theorem}

\begin{proof}
	Let us first suppose that we are given $s\geq n+i$, an upper bound on the size 
	of a circuit computing $Y_i$ from $X,Y^{1:i-1}$. Our algorithm follows the 
	approach of Bshouty et al., adapted to the relational synthesis setting.
	
	The algorithm proceeds iteratively. We maintain a set of counterexamples 
	$\{(\vec{x}_1,\vec{y}_1),\ldots,(\vec{x}_k,\vec{y}_k)\}$ where each 
	$(\vec{x}_j,\vec{y}_j)$ is a complete assignment to all variables $(X,Y)$ 
	satisfying $F(\vec{x}_j,\vec{y}_j) = 1$. Initially, we have $k=0$ 
	counterexamples.
	
	At iteration $k+1$, we perform the following steps:
	
	\textbf{Step 1: Sample candidate circuits.} Using the NP oracle, we sample 
	$d\cdot s$ circuits $g_1,\ldots,g_{ds}$ (for some constant $d>1$) that are 
	$(1+\delta)$-close to uniformly at random (in statistical distance) from the set of all circuits of size 
	at most $s$ that are \emph{consistent with all previous counterexamples}. 
	Specifically, each sampled circuit $g_j$ must satisfy: for every counterexample 
	$(\vec{x}_\ell,\vec{y}_\ell)$ with $\ell \in \{1,\ldots,k\}$, we require 
	$g_j(\vec{x}_\ell,\vec{y}_\ell^{1:i-1}) = \vec{y}_\ell[i]$, where 
	$\vec{y}_\ell[i]$ denotes the $i$-th component of $\vec{y}_\ell$.
	
	\textbf{Step 2: Construct majority vote function.} We define $h_{k+1}(X,Y^{1:i-1})$ by
	\begin{align*}
		\mathrm{majority}\{g_1(X,Y^{1:i-1}),\ldots,
		g_{ds}(X,Y^{1:i-1})\}.
	\end{align*}
	
	\textbf{Step 3: Check for counterexample.} We use the NP oracle to determine 
	whether there exists an assignment $(\vec{x}_{k+1},\vec{y}_{k+1})$ satisfying 
	\begin{align*}
		E_{k+1}(X,Y) := F(X,Y) \wedge (Y_i \neq h_{k+1}(X,Y^{1:i-1})).
	\end{align*}
	
	If no such assignment exists (i.e., $E_{k+1}$ is unsatisfiable), then 
	$h_{k+1}(\vec{x},\vec{y}^{1:i-1}) = Y_i$ for all $(\vec{x},\vec{y})$ 
	satisfying $F$, and we output $h_{k+1}$ as the desired Skolem function.
	
	If such an assignment exists, we add $(\vec{x}_{k+1},\vec{y}_{k+1})$ to our 
	counterexample set and proceed to iteration $k+2$.
	
	\textbf{Analysis.} We now analyze why this algorithm terminates in polynomial 
	time. The key insight is that $h_{k+1}$ is a ``good'' hypothesis: either it correctly computes $Y_i$, or any counterexample 
	to it eliminates a large fraction of the remaining candidate circuits.
	For a suitable choice of constants $d$ and $\delta$, and for each input 
	$(\vec{x},\vec{y}^{1:i-1}) \in \{0,1\}^n \times \{0,1\}^{i-1}$, the 
	probability that $h_{k+1}$ agrees with fewer than $1/4$ of the circuits of 
	size at most $s$ that are consistent with all previous counterexamples is at 
	most $2^{-2s}$. 
	
	Since there are $2^{n+i-1}$ possible inputs $(\vec{x},\vec{y}^{1:i-1})$ and 
	$s \geq n+i$, we have $2^{n+i-1} \leq 2^{s-1}$. By a union bound, the 
	probability that there exists any input on which $h_{k+1}$ agrees with fewer 
	than $1/4$ of the consistent circuits is at most 
	$2^{n+i-1} \cdot 2^{-2s} \leq 2^{s-1-2s} = 2^{-s-1}$.
	
	Therefore, with probability at least $1-2^{-s-1}$, the function $h_{k+1}$ is 
	$1/4$-good: for any counterexample $(\vec{x},\vec{y})$ where 
	$h_{k+1}(\vec{x},\vec{y}^{1:i-1}) \neq \vec{y}[i]$, at least $1/4$ of the 
	circuits consistent with previous counterexamples will disagree with the true 
	Skolem function on this input, and thus be eliminated.
	
	Since there are at most $2^{O(s\log s)}$ circuits of size at most $s$, and 
	each good iteration eliminates at least a $1/4$ fraction of remaining circuits, 
	the algorithm terminates within $O(s\log s)$ iterations. By a union bound over 
	all iterations, the probability that we find a good $h_\ell$ in all 
	$O(s\log s)$ iterations is at least 
	$1 - O(s\log s) \cdot 2^{-s-1} \geq 1 - 2^{-s+\log s+\log\log s+D}$ 
	for some constant $D$.
	
	\textbf{Handling unknown circuit size.} If the circuit size $s$ is not known 
	in advance, we start with an initial guess $s_0$ and double it repeatedly until 
	the algorithm succeeds. The total probability of failure across all doubling 
	phases is 
	\begin{align*}
		\sum_{j=0}^\infty 2^{-s_0 \cdot 2^j + j + \log s_0 + \log(j+\log s_0) + D} 
		< \frac{1}{4}
	\end{align*}
	for sufficiently large $s_0$.
	\end{proof}

\noindent

An NP oracle is actually not even necessary for efficient synthesis of Skolem functions for relations on a small number of $Y$ variables---we can obtain a Skolem function from the specification itself:
\begin{lemma}\label{lem:pi2-det-small-skolem}
	There is a deterministic algorithm that, given  $F(X,Y)$ where 
	$m = |Y|$, returns a circuit of size $O(|F|m\cdot 2^{2m})$ in time
	polynomial in $|F|$ and $2^m$ that computes a Skolem function for $F$.
\end{lemma}
\begin{proof}
	Observe indeed that we can compute the lexicographically first $Y$ satisfying
	$F$ for a given $X$ by
	\[
	f_i(X)=\bigvee_{{\substack{\vec{b}\in\{0,1\}^m:\\b_i=1}}}F(X,\vec{b})\land\bigwedge_{{\substack{\vec{b}'\in\{0,1\}^m:\\\vec{b}'<\vec{b}}}}\neg F(X,\vec{b}').
	\]
	Indeed, since this circuit outputs a $\vec{y}$ for a given $\vec{x}$ satisfying $F(\vec{x},\vec{y})$
	whenever one exists, it is a Skolem function and can be produced in
	time polynomial in the size of the circuit.
\end{proof}

\noindent
Observe that what was critical in the above construction was that the number 
of possible solutions was small. We now show that with the NP oracle, this 
approach can be extended to generate Skolem functions in time polynomial in 
the size of the smallest image among all valid $\Psi$.

\begin{theorem}
	Let $F(X,Y)$ be a relational specification. If there exists an ordered set of 
	Skolem functions $\Psi^*$ for $F$ with $|\mathrm{Im}(\Psi^*)| = k$, then 
	there is a randomized algorithm with access to an NP oracle that synthesizes 
	Skolem functions for $F$ in time $\poly(n, m, |F|, k)$.
\end{theorem}

\begin{proof}
	Since there exist Skolem functions $\Psi^*$ with image size $k$, we know 
	that $S = \mathrm{Im}(\Psi^*) \subseteq \{0,1\}^m$ satisfies $|S| = k$ and
		$\satisfying{F}_{\downarrow X} = \{\vec{x}: \exists \vec{y} \in S \text{ s.t. } 
		F(\vec{x},\vec{y}) = 1\}.$

	We construct Skolem functions by first identifying a small subset 
	$S' \subseteq \{0,1\}^m$ that covers all satisfying inputs, then building 
	circuits that select appropriate outputs from $S'$.
	
	\textbf{Phase 1: Constructing the covering set $S'$.}
	We approximate the greedy set-cover algorithm using the NP oracle. Initially, 
	let $S_0 = \emptyset$. We iteratively build sets $S_i$ until
	\begin{align*}
		\satisfying{F}_{\downarrow X} = \{\vec{x}: \exists \vec{y} \in S_i 
		\text{ s.t. } F(\vec{x},\vec{y}) = 1\}.
	\end{align*}
	
	In iteration $i$, we first estimate the number of inputs covered by $S_i$. 
	Using the NP oracle on the formula $\bigvee_{\vec{y} \in S_i} F(X,\vec{y})$ 
	with appropriate hashing, we obtain $N_i$ such that
		$N_i \geq \frac{1}{2} \cdot |\{\vec{x}: \exists \vec{y} \in S_i \text{ s.t. } 
		F(\vec{x},\vec{y}) = 1\}|$
	
	Next, we select $\vec{y}_i \in \{0,1\}^m \setminus S_i$ to add to $S_{i+1}$. 
	We use the NP oracle on $
		F(X,Y) \wedge \bigwedge_{\vec{y} \in S_i} \neg F(X,\vec{y})$
	combined with a $\lceil\log(N_i/(2k))\rceil$-bit hash on $X$ to ensure we 
	sample from uncovered inputs. This yields $\vec{y}_i$ such that, with 
	probability at least $1/2$,
	\begin{align*}
		|\{\vec{x}: F(\vec{x},\vec{y}_i) = 1 \text{ and } \forall \vec{y} \in S_i, 
		F(\vec{x},\vec{y}) = 0\}| \geq \frac{N_i}{2k}.
	\end{align*}
	
	Since $S$ has size $k$ and covers all inputs in $\satisfying{F}_{\downarrow X}$, 
	at least one element of $S$ must cover at least a $1/k$ fraction of the 
	remaining uncovered inputs. Thus, with constant probability, each iteration 
	reduces the number of uncovered inputs by a factor of $(1 - 1/(4k))$.
	After $O(k \log n)$ iterations, with high probability, $S_i$ covers all 
	inputs in $\satisfying{F}_{\downarrow X}$ and at that point we have $S' = S_i$. The expected size of $S'$ is 
	$O(k \log n)$ by the analysis of the greedy set-cover algorithm. 
	To handle the unknown value of $k$, we run the algorithm with successively 
	doubling guesses for $k$, starting from $1$. The total time remains 
	polynomial in $k$.
	
	\textbf{Phase 2: Constructing Skolem functions.}
	Given the covering set $S'$ with $|S'| = O(k \log n)$, we construct Skolem 
	functions that, for each input $\vec{x}$, output the lexicographically first 
	$\vec{y} \in S'$ such that $F(\vec{x}, \vec{y}) = 1$.
	Using the construction from Lemma~\ref{lem:pi2-det-small-skolem}, we build 
	a circuit computing:
	\begin{align*}
		\psi_i(\vec{x}) = \bigvee_{\substack{\vec{y} \in S': y_i = 1}} 
		\left( F(\vec{x}, \vec{y}) \wedge \bigwedge_{\substack{\vec{y}' \in S': 
				\vec{y}' <_{\text{lex}} \vec{y}}} \neg F(\vec{x}, \vec{y}') \right).
	\end{align*}
	
	This circuit has size $O(|F| \cdot m \cdot |S'|^2) = O(|F| \cdot m \cdot 
	k^2 \log^2 n)$ and can be constructed in time polynomial in $n$, $m$, $|F|$, 
	and $k$.
\end{proof}

\section{Conclusion}\label{sec:conclusion}  

We have presented a systematic theoretical investigation into the power of 
functional synthesis approaches that rely on NP oracles. Our main 
contributions are fourfold. First, we examined the applicability of 
computational learning theory to functional synthesis and demonstrated 
fundamental obstacles that prevent direct extension of these techniques due 
to the relational nature of specifications. Second, we established 
fundamental limitations of interpolation-based approaches for handling 
uniquely-defined variables, proving that resolution-based interpolation must 
produce exponential-size circuits even when small Skolem functions exist. 
Third, we proved that access to an NP oracle is inherently necessary for 
efficient synthesis, suggesting that the reliance on SAT solvers in practical 
tools is not merely an implementation choice but rather necessary for 
achieving scalability. Fourth, we showed that NP oracles enable synthesis 
of Skolem functions in time that scales with the specification size and 
function image size.

\section*{Acknowledgements}
Juba's work was supported in part by the NSF award IIS-1942336.
Meel's work was supported in part by the Natural Sciences and
Engineering Research Council of Canada (NSERC) [RGPIN-
2024-05956]

\end{document}